\documentclass{llncs}

\usepackage[english]{babel}
\usepackage{amsmath,amssymb}
\usepackage{subfigure}

\usepackage{tikz}
\usetikzlibrary{automata,patterns,topaths,shapes}
\usetikzlibrary{calc}
\usepackage{gastex}

\makeatletter
\def\moverlay{\mathpalette\mov@rlay}
\def\mov@rlay#1#2{\leavevmode\vtop{%
  \baselineskip\z@skip \lineskiplimit-\maxdimen
   \ialign{\hfil$\m@th#1##$\hfil\cr#2\crcr}}}
\newcommand{\charfusion}[3][\mathord]{
    #1{\ifx#1\mathop\vphantom{#2}\fi
        \mathpalette\mov@rlay{#2\cr#3}
      }
    \ifx#1\mathop\expandafter\displaylimits\fi}
\makeatother

\newcommand{\cupdot}{\charfusion[\mathbin]{\cup}{\cdot}}
\newcommand{\bigcupdot}{\charfusion[\mathop]{\bigcup}{\cdot}}

\newcommand{\La}{\mathcal{L}}
\newcommand{\Obs}{\mathcal{O}}
\newcommand{\T}{\mathcal{T}}
\newcommand{\fee}{\varphi}
\newcommand{\eps}{\varepsilon}

\newcommand{\dom}{\operatorname{dom}}

\newcommand{\id}{\operatorname{Id}}

\newcommand{\trans}[1]{\mbox{$\stackrel{#1}{\longrightarrow}$}}

\colorlet{auxdraw}{blue!25!black!70}
\colorlet{auxfill}{auxdraw!30}
\colorlet{auxtext}{black}
\tikzstyle{aux}=[draw=auxdraw,fill=auxfill,text=auxtext]
\tikzstyle{phee}=[draw=auxdraw,fill=auxfill]
\tikzstyle{every state}=[aux]

\title{Verification of Information Flow Properties \\under Rational Observation}

\author{B\'eatrice B\'erard\inst{1}\fnmsep \thanks{Supported by a
grant from Coop\'eration France-Qu\'ebec, Service Coop\'eration et
Action Culturelle 2012/26/SCAC.} \and John Mullins\inst{2}\fnmsep \thanks{supported by the NSERC Discovery Individual grant No. 13321 (Government of Canada), the FQRNT Team grant No. 167440 (Quebec's Government) and the CFQCU France-Quebec Cooperative grant No. 167671 (Quebec's Government).}}

\institute{Sorbonne Universit\'e, Universit\'e Pierre \& Marie Curie
 LIP6/MoVe, CNRS UMR 7606, Paris, France \and 
 \'Ecole Polytechnique de Montr\'eal,
 Campus of  the Universit\'e de Montr\'eal, Montreal (Quebec), Canada. \\
\email{Beatrice.Berard@lip6.fr,john.mullins@polymtl.ca}}

\begin{document}

\maketitle

\begin{abstract} Information flow properties express the
capability for an agent to infer information about secret behaviours
of a partially observable system. In a language-theoretic setting,
where the system behaviour is described by a language, we define the
class of rational information flow properties (RIFP), where observers
are modeled by finite transducers, acting on languages in a given
family $\La$. This leads to a general decidability criterion for the
verification problem of RIFPs on $\La$, implying PSPACE-completeness
for this problem on regular languages. We show that most trace-based
information flow properties studied up to now are RIFPs, including
those related to selective declassification and conditional
anonymity. As a consequence, we retrive several existing decidability
results that were obtained by ad-hoc proofs. 

\keywords{Information flow, Security predicates, Opacity,
Declassification, Conditional anonymity, Rational transducers,
Formal verification.}
\end{abstract}

\section{Introduction}

\paragraph{Motivations.}
Generic models for information flow properties aim at expressing, in a
uniform setting, the various capabilities of observers to infer
information from partially observable systems. These models provide a
description of the system behaviour, a parametric description of the
observation by the environment and the secret parts of the system, and
a security criterion. A security property is an instantiation of such
a model, with the goal of avoiding a particular information flow.
Generic models have been thoroughly investigated, for instance
in~\cite{Mantel00,focardi01classification,Mazare2008}. They propose
various classifications and comparisons of security properties, either
for transition systems or directly for traces. In the case of
transition systems~\cite{focardi01classification,Mazare2008}, the
branching structure permits to express security properties as
equivalences like weak or strong (bi-)simulations. For trace-based
models, properties are stated as relations between languages, also
called security predicates in~\cite{Mantel00}.

In addition to classification, an important question about security
properties concern their verification: given a system $S$ and a
security property $P$, does $S$ satisfy $P$\,?
Since~\cite{focardi01classification}, much attention has been given to
such questions for various classes of systems (or their sets of
traces) and security
properties~\cite{Mazare2008,DSouzaHRS11,Cassez2012,SecLTL2012,best-darondeau2012,mullins2014,ClarksonFKMRS14}.
This is the problem we consider in this work, for a subclass of 
trace-based information flow properties.
%
%
%
\paragraph{Contributions.}  
We first introduce the class of Rational Information Flow Properties
(RIFP), in a language-theoretic setting. In this class, observations
are modeled by rational transducers, called here rational observers.
For a language $L$ in some family of languages $\La$, an RIFP is then
defined as an inclusion relation $L_1 \subseteq L_2$, where $L_1$ and
$L_2$ are obtained from $L$ by inductively applying rational
observers, unions and intersections. This mechanism produces the set
of properties $RIF(\La)$, and a generic decidability result can be
stated for the verification problem of these properties.  In the
particular case of the family $\mathcal{R}eg$ of regular languages,
generated by finite automata (also called labelled transition
systems), we obtain a PSPACE-complete verification problem for the
class $RIF(\mathcal{R}eg)$.  We then proceed to show that this result
subsumes most existing decidability results for security properties on
regular languages, thus establishing the pertinency of our model.
This involves expressing properties in our formalism by designing
suitable rational observers.  We first consider the particular case
where observations are functions and we show that opacity properties
with regular secrets are RIFPs. To illustrate the expressiveness of
RIFPs, we introduce a subclass of functional rational observers that
we call rational Orwellian observers and show that several properties
including {\em intransitive non-interference} and {\em selective
  intransitive non-interference} for a language $L \in \La$ are in
$RIF(\La)$. We also reduce their verification to the verification of
opacity w.r.t. Orwellian observers. These observers are more powerful
than those considered so far in literature as they model not only
observers constrained to a fixed {\em a priori} interpretation of
unobservable events (static observers) or even to observers able to
base this interpretation on observation of previous events (dynamic
observers), but also able to re-interpret past unobservable events on
the base of subsequent observation.  We finally consider general
observers and we show that all Mantel's Basic Security Predicates
(BSPs) are RIFPs. Finally, we illustrate the applicability of our
framework by providing the first formal specification for {\em
  conditional anonymity} guaranteeing anonymity of agents unless
revocation (for instance, the identity of an agent discovered to be
dishonest can be revealed).

\paragraph{Outline.} The rest of the paper is organized as
follows. Rational Information flow properties are defined in
Section~\ref{observation}, with the associated decidability results.
RIFPs w.r.t. rational observation functions are investigated
in Section~\ref{S-Opacity}: rational opacity properties as RIFP are
presented in~\ref{ratop}, Orwellian observers in
\ref{sec:Orwellian} and their application to {\em intransitive
  non-interference} and {\em selective intransitive non-interference}
in~\ref{section:SD}.  RIFPs w.r.t. general rational observation
relations are investigated in~\ref{sec:ex}: BSPs as RIFPs are
presented in~\ref{sec:bsp} and an application of general rational
observation relation to conditional anonymity is presented
in~\ref{sec:CA}.  In Section~\ref{sec:relw}, we discuss related work
and we conclude in Section~\ref{sec:conc}.  

\section{Rational Information flow properties}
\label{observation}

We briefly recall the notions of finite automata and finite
transducers before defining rational information flow properties.

\subsection{Automata and transducers}
The set of natural numbers is denoted by $\mathbb{N}$ and the set of
\emph{words} over a finite alphabet $A$ is denoted by $A^*$, with
$\varepsilon$ for the empty word and $A^+ = A^* \setminus
\{\varepsilon\}$. The length of a word $w$ is written $|w|$ and for
any $a \in A$, $|w|_a$ is the number of occurrences of $a$ in $w$. A
\emph{language} is a subset of~$A^*$. 

\paragraph{Finite Labelled Transition Systems.}  A finite labelled
transition system (LTS or automaton for short), over a finite set
$Lab$ of labels, is a tuple $\mathcal{A} =\langle
Q,I,\Delta,F\rangle$, where $Q$ is a finite set of states, $I
\subseteq Q$ is the subset of initial states, $\Delta \subseteq Q
\times Lab \times Q$ is a finite transition relation and $F \subseteq
Q$ is a set of final states.  Note that $Lab$ can be an alphabet but
also a (subset of a) monoid.

Given two states $q, q' \in Q$, a \emph{path} from $q$ to $q'$ with
\emph{label} $u$, written as $q \trans{u} q'$, is a sequence of
transitions $q \trans{a_1} q_1$, $q_1 \trans{a_2} q_2, \cdots, \ q_{n-1}
\trans{a_n} q'$, with $a_i \in Lab$ and $q_i \in Q$, for $1 \leq i
\leq n-1$, such that $u=a_1 \cdots a_n$. The path is \emph{accepting}
if $q \in I$ and $q' \in F$, and the language of $\mathcal{A}$,
denoted by $L(\mathcal{A})$, is the set of labels of accepting paths.
A regular language over an alphabet $A$ is a subset of $A^\ast$
accepted by a finite LTS over the set of labels $A$.

\medskip \noindent \textbf{Finite Transducers.}  A finite transducer
(or transducer for short) is a finite LTS $\T$ with set of
labels $Lab \subseteq A^\ast \times B^\ast$ for two alphabets $A$ and
$B$.  A label $(u,v) \in A^\ast \times B^\ast$ is also written as
$u|v$.  The subset $L(\T)$ of $A^\ast \times B^\ast$ is a
\emph{rational relation}~\cite{sakarovitch09} from $A^\ast$ to
$B^\ast$. The transducer $\T$ is said to realize the relation $L(\T)$
(see Fig.~\ref{fig:ex1} for basic examples of transducers).

Given a rational relation $R$, we write $R(u)= \{v \in B^\ast \mid
(u,v)\in R\}$ for the image of $u \in A^\ast$, $R^{-1}(v)= \{u \in A^*
\mid (u,v) \in R\}$ for the inverse image of $v \in B^\ast$, possibly
extended to subsets of $A^\ast$ or $B^\ast$ respectively, and $\dom(R)=\{u
\in A^\ast \mid \exists v \in B^\ast, (u,v)\in R\}$ for the domain of
$R$.
The relation $R$ is complete if $\dom(R)=A^*$, it is a function if for
each $u \in \dom(R)$, $R(u)$ contains a single element $v \in B^*$.

For a subset $P$ of $A^\ast$, the identity relation $\{(u,u) \mid u
\in P\}$ on $A^\ast \times A^\ast$ is denoted by $\id_P$.  The
composition of rational relations $R_1$ on $A^\ast \times B^\ast$ and
$R_2$ on $B^\ast \times C^\ast$, denoted by $R_1 R_2$ (from left to
right) or by $R_2 \circ R_1$ (from right to left), is the rational
relation on $A^\ast \times C^\ast$ defined by $\{(u,w) \mid \exists v
\ (u,v) \in R_1 \land (v,w) \in R_2\}$ (\cite{elgot65}).  The family
of regular languages is closed under rational
relations~\cite{berstel79}.

\subsection{Rational observers}

Information flow properties are related to what an agent can learn
from a given system. In a language-based setting, the behavior of the
system is described by a language $L$ over some alphabet $A$, and some
function $\Obs$ associates with each $w\in L$ its observation
$\Obs(w)$ visible by the agent. 
We generalize the notion of
observation by defining $\Obs$ as a relation on $A^* \times B^*$ for
some alphabet $B$, but we restrict $\Obs$ to be a rational relation.

\begin{definition}[Rational observer] \label{def:ratobserver}
  A rational observer is a rational relation $\Obs$ on $A^* \times
  B^*$, for two alphabets $A$ and $B$. The observation of a word $w
  \in A^*$ is the set $\Obs(w)= \{w' \in B^* \mid (w,w') \in \Obs \}$
  and for any language $L \subseteq \dom(\Obs)$, the observation of
  $L$ is $\Obs(L)=\cup_{w\in L} \Obs(w)$.
\end{definition}

As pointed out in~\cite{DSouzaHRS11}, a large amount of information
flow properties of a language $L$ are expressed as relations of the
form $op_1(L) \subseteq op_2(L)$, for some language theoretic
operations $op_1$ and $op_2$.
Actually, we show below that $op_1$ and $op_2$ are often rational
relations corresponding to some specific observations of $L$.  Also,
we define the class of rational information flow properties as those
using rational observers, and positive boolean operations:
\begin{definition}[Rational information flow property]
  A rational information flow property (RIFP) for a language $L$ is
  any relation of the form $L_1 \subseteq L_2$, where $L_1$ and $L_2$
  are languages given by the grammar:
  \[L_1, L_2 ::=~L~|~\Obs(L_1)~|~ L_1\cup L_2~|~ L_1\cap L_2\] where
  $\Obs$ is a rational observer.
\end{definition} 

Hence, from Def.~\ref{def:ratobserver}, we recover information flow
properties of $L$ of the form $\Obs_1(L) \subseteq \Obs_2(L)$ for two
rational observers, as a particular case. However it has to be noted
that Def.~\ref{def:ratobserver} does not reduce to these inclusions
since rational relations are not closed under
intersection~\cite{berstel79}. Given a family of languages $\La$, we
define $RIF(\La)$ as the set of RIFPs for languages in $\La$. We
immediately have the following general result:
\begin{proposition}\label{prop:decidability}
  Let $\La$ be a family of languages closed under union, intersection,
  and rational transductions, such that the relation $\subseteq$ is
  decidable in $\La$.  Then any property in $RIF(\La)$ is decidable.
\end{proposition}
In particular, the class $\mathcal{R}eg$ of regular languages
satisfies the conditions above, with a PSPACE-complete inclusion
problem. We then have:
\begin{corollary}\label{cor:regularcase}
  The problem of deciding a property in $RIF(\mathcal{R}eg)$ is
  PSPACE-complete.
\end{corollary} 
\begin{proof}
  It follows from the remark above that the problem is in PSPACE.  For
  PSPACE-hardness, recall that for a language $K$, the relation
  $\Obs_K$ defined by $\Obs_K(w)=\{w\} \cap K$ is a rational observer
  if (and only if) $K$ is a regular language~\cite{sakarovitch09}. Let
  $L_1$ and $L_2$ be two regular languages, and let $\Obs_{L_1}$,
  $\Obs_{L_2}$ be the two corresponding relations, then for $L=A^*$,
  we have $L_1 \subseteq L_2$ if and only if $\Obs_{L_1}(L) \subseteq
  \Obs_{L_2}(L)$.\qed
\end{proof}
This corollary subsumes many existing decidability results for IF
properties. The rest of the paper is devoted to establish reductions
of some of these to the $RIF(\mathcal{R}eg)$ verification problem.

\section{RIF properties with rational functions}\label{S-Opacity}
In this section, we consider the generic model of opacity introduced
in~\cite{Mazare2008} for transition systems.  Opacity is parametrized
with observation functions, that are classified in~\cite{Mazare2008}
as {\em static}, {\em dynamic} or {\em Orwellian} to reflect the
computational power of the observer. In a static observation, actions
are always interpreted in the same way. It is defined as a morphism
and hence, it is a rational function. A particular case of static
observer is the projection $\pi_B$ from $A^*$ into $B^*$ for a
subalphabet $B$ of $A$, so that $\pi_B(a)= a$ if $a\in B$ and
$\pi_B(a)= \eps$ otherwise. In a dynamic observation function,
interpretation of the current action depends on the sequence of
actions observed so far and hence, it is also a rational function.

\begin{example}
  In Fig.~\ref{fig:ex1} (where all states are final states), the left
  hand side depicts a transducer realizing the projection from
  $\{a,b\}^*$ onto $\{b\}^*$ while the right hand side depicts a
  transducer realizing the following dynamic observation function
  (translated from~\cite{Cassez2012}): The first occurrence of the
  first action is observed, then nothing is observed until the first
  occurrence of the second action ($b$ if the trace begins with $a$
  and $a$ otherwise) and everything is observed in clear as soon as
  this second action occurs that is, $\Obs(aa^\ast bu) = abu$ and
  $\Obs(bb^\ast au) = bau$ for any $u \in \{a, b\}^\ast$.
\end{example}

\begin{figure}[htbf]
\unitlength=0.8mm
\gasset{loopdiam=6,Nh=8,Nw=8,Nmr=4}
\begin{picture}(80,45)(-15,-4)
\node[linecolor=blue!70,fillcolor=blue!20](s0)(30,20){$0$}
\imark(s0)
\drawloop(s0){$a|\eps,b|b$}

\put(80,0){\node[linecolor=blue!70,fillcolor=blue!20](q0)(10,30){$0$}
\imark(q0)
\node[linecolor=blue!70,fillcolor=blue!20](q1)(30,30){$1$}
\node[linecolor=blue!70,fillcolor=blue!20](q3)(50,30){$3$}
\node[linecolor=blue!70,fillcolor=blue!20](q2)(30,10){$2$}

\drawedge(q0,q1){$a|a$} \drawedge[ELside=r](q0,q2){$b|b$}
\drawloop[loopangle=270](q2){$b|\eps$}
\drawedge[ELside=r](q2,q3){$a|a$}
\drawloop(q1){$a|\eps$}
\drawedge(q1,q3){$b|b$}
\drawloop(q3){$a|a,b|b$}
}
\end{picture}
\caption{Examples of transducers realizing basic observation
  functions}
\label{fig:ex1}
\end{figure}
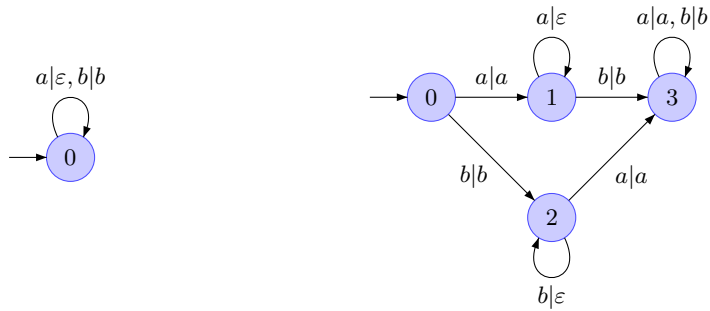

In Orwellian observation functions, the current observation depends
not only on the prefix of actions observed so far but also on the
complete trace. It reflects the capability of the observer to use
subsequent knowledge to re-interpret past actions. In the rest of this
section we will study opacity w.r.t.\ rational Orwellian observers.

\subsection{Opacity w.r.t. rational functions} \label{ratop}
In its original setting, opacity is related to a language $L \subseteq
A^\ast$ modelling the behaviour of a system, a function $\Obs$ from $A^*$ to
$B^*$ and in addition, a predicate $\varphi$ given as a subset
of $L$, describing a secret. Two words $w$ and $w'$ of $L$ are
observationally equivalent for $\Obs$ if $\Obs(w) = \Obs(w')$. The
\emph{observation class} of $w$ in $L$ is the set $[w]_{\cal O}^L=\{
w' \in L \mid \Obs(w) = \Obs(w')\}=L \cap \Obs^{-1}(\Obs(w))$.

The secret $\fee$ is opaque in $L$ for $\Obs$ if for any word in
$\fee$, there is another word in $L\setminus \fee$ such that $w$ and
$w'$ are observationally equivalent. Hence, $\fee$ is opaque if and
only if ${\cal O}(\fee) \subseteq {\cal O}(L \setminus \fee)$, which
we take as definition when $\Obs$ is a rational function:

\begin{definition}[Rational Opacity]
  \label{def:opacity-disclosure}
  Given a language $L \subseteq A^\ast$, a language $\fee \subseteq L$
  and a rational function $\Obs$, $\fee$ is rationally opaque in $L$
  for $\Obs$ if ${\cal O}(\fee) \subseteq {\cal O}(L \setminus \fee)$.
\end{definition}

 
The information flow deduced by an observer when the system is not
opaque is captured by the notion of secret disclosure: A word $w\in L$
discloses the secret $S$ w.r.t. $\Obs$ if $[w]_{\cal O}^{L} \subseteq
\fee$. We have:
\begin{proposition}
  Rational opacity properties on languages in some family $\La$ for
  regular secrets belong to $RIF(\La)$.
\end{proposition}
\begin{proof}
  As already seen in the proof of Corollary~\ref{cor:regularcase},
  intersection with a regular set $K$ is a rational observation
  $\Obs_K$. Since the secret $\fee$ is regular, opacity of $\varphi$
  in $L$ for $\Obs$ is equivalent to $\Obs (\Obs_\fee(L)) \subseteq
  \Obs (\Obs_{\neg \fee}(L))$.\qed
\end{proof}

Non-interference and weak and strong anonymity have been shown to
reduce to opacity w.r.t. suitable observers (see \cite{Mazare2008}).
In~\cite{Cassez2012}, PSPACE-hardness is established for opacity of
regular secrets for regular languages w.r.t. static and dynamic
observers.

\subsection{Rational Orwellian observers}\label{sec:Orwellian}
In the sequel, we denote the disjoint union by $\cupdot$.  In our
context, Orwellian observation functions from~\cite{Mazare2008} are
realized by rational Orwellian observers:

\begin{definition}[Rational Orwellian Observer] \label{def:OObs} A
  {\em rational Orwellian observer} is a rational function, given as a
  disjoint union of functions: $\Obs = \cupdot_{1 \leq i \leq n}
  \Obs_i$, where the domains $\{\dom(\Obs_i), 1 \leq i \leq n\}$ form
  a partition of $A^*$. The partial functions $\Obs_i$ are called
  \emph{views}.
\end{definition}

Note that $\Obs$ is a function because the domains of the views are
disjoint.  We simply call these functions Orwellian observers for
short, since there is no ambiguity in our context. The terminology
\emph{Orwellian} comes from the ability of the observer to somehow see
in the future, as illustrated in the following example.

\begin{example}[A simple example]
\label{ex:simple}
The function $\Obs = \Obs_a \cupdot \Obs_b \cupdot \Obs_\eps$ is an
Orwellian observer on $\{a, b\}$ realized by the transducer depicted
in Fig.~\ref{fig:OPex1}. The function is defined by $\Obs(\eps)=\eps$
and:

\[
\Obs(w) = \left\{ \begin{array}{l@{ }l}
\pi_{\{b\}}(w)  & 
  \mbox{ if the last letter of }  w  \mbox{ is } a\\
  \pi_{\{a\}}(w) &
   \mbox{ if the last letter of }  w  \mbox{ is } b.
\end{array}\right.
\]
Hence, the observer interpretation of the current event depends on the
last event of the trace. If it is $a$ then ${\cal O}$ interprets the
trace as its projection over $\{b\}$ and the other way around, if it
is $b$ then it interprets the trace as its projection over $\{a\}$.

\begin{figure}[ht]
  \begin{center}
\unitlength=0.8mm
\gasset{loopdiam=6,Nh=8,Nw=8,Nmr=4}
    \begin{picture}(130,20)(-10,0)
    \node[linecolor=blue!70,fillcolor=blue!20](A0)(-10,5){$p_0$}
    \nodelabel[ExtNL=y,NLdist=8,AHnb=0,ilength=-2,NLangle= 180](A0){$\Obs_a :$} 
    \imark(A0) 
    
    \node[linecolor=blue!70,fillcolor=blue!20](A1)(15,5){$p_1$}
\fmark[fangle=270](A1)

\node[linecolor=blue!70,fillcolor=blue!20](A3)(55,5){$q_0$}
 \nodelabel[ExtNL=y,NLdist=8,AHnb=0,ilength=-2,NLangle=  180](A3){$\Obs_b:$} 
          \imark(A3)
                   
\node[linecolor=blue!70,fillcolor=blue!20](A4)(80,5){$q_1$} 
\fmark[fangle=270](A4)
                               
 \node[linecolor=blue!70,fillcolor=blue!20](A6)(120,5){$r_0$}
 \nodelabel[ExtNL=y,NLdist=8,AHnb=0,ilength=-2,NLangle=  180](A6){
$\Obs_\eps :$} 
\imark(A6) \fmark[fangle=270](A6)
          
    \drawedge(A0,A1){$a|\eps$}
      \drawloop[loopangle=90](A0){$a|\eps, \ b|b$}
            \drawloop[loopangle=90](A3){$a|a, \ b|\eps$}
            \drawedge(A3,A4){$b|\eps$}
    \end{picture}
  \end{center}
  \caption{The Orwellian observer ${\cal O} = \Obs_a \cupdot \Obs_b
    \cupdot \Obs_\eps$.}
\label{fig:OPex1}
\end{figure}

Despite its observational power, this observer is not able to deduce
whether the first event in the trace in $L = (a+b)(a^\ast +
b^\ast)(a+b)$ is an $a$. Indeed, let $\varphi = a(a^\ast +
b^\ast)(a+b)$ be the secret, corresponding to the set of traces in $L$
with $a$ as the first event. Then $\varphi$ is opaque w.r.t. $\Obs$ in
$L$. To see this, if a secret trace $w$ is observed, examine what
$\cal O$ can deduce from this observation.
\begin{itemize}
\item If $w$ ends with an $a$ then $\Obs(w) = b^n$ for some $n
  \geq 0$ but $b^na \not \in \varphi$ is also observed by $b^n$.
\item If $w$ ends with a $b$ then $\Obs(w) = a^n$ for some $n
  \geq 0$ but $ba^nb \not \in \varphi$ is also observed by $a^n$.
\end{itemize}

\end{example}

\begin{example}[Static and dynamic observers]
  Static and dynamic observations are of course special cases of
  Orwellian observers, where $\Obs$ consists of a single complete
  view. Note that static and dynamic observations preserve prefixes
  while it is not necessarily the case for Orwellian observations (see
  examples~\ref{ex:simple} and~\ref{ex:pi_{o,d}}).
\end{example}

\begin{example}[Intransitive non-interference]
\label{ex:pi_{o,d}}
Let $A = V \cupdot C \cupdot D$ be a partition of the alphabet into
visible actions in $V$, confidential actions in $C$ and
declassification actions in $D$. When a declassification action occurs
in a word, the prefix is observed in clear. The corresponding
observation function is called in~\cite{mullins2014} the
\emph{projection on $V$ unless $D$}, and defined as a mapping $
\pi_{V,D} : A^* \rightarrow A^*$ such that $\pi_{V,D} (\epsilon) =
\epsilon$ and
$$
\pi_{V,D} (u a) = \left\{ \begin{array}{l@{ }l} 
u a \ & \ \mbox{if} \ a \in D, \\
\pi_{V,D} (u) a \ & \ \mbox{if} \ a \in V, \\
    \pi_{V,D} (u) & \ \mbox{otherwise.}
\end{array}\right.
$$ 
A language $L$ satisfies \emph{intransitive non-interference (INI)} if
$\pi_{V,D} (L) \subseteq L$. Again:

\begin{proposition}
  The function $\pi_{V,D}$ is an Orwellian observer, hence INI for
  languages in $\La$ belongs to $RIF(\La)$.
\end{proposition}
\begin{proof}
  The function $\pi_{V,D}$ is a sum of two views: $\pi_{V,D}=
  \Obs_{\eps} \cupdot \Obs_D$, realized by the transducers depicted in
  Fig.~\ref{fig:pi_{o,d}}.\qed
\end{proof}

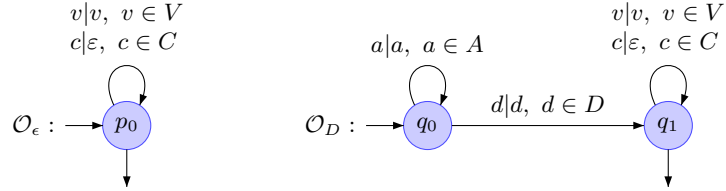
\begin{figure}[ht]
  \begin{center}
    \unitlength=0.8mm
\gasset{loopdiam=6,Nh=8,Nw=8,Nmr=4}
    \begin{picture}(90, 30)(-15,0)
\node[linecolor=blue!70,fillcolor=blue!20](A0)(-10,5){$p_0$}
\nodelabel[ExtNL=y,NLdist=8,AHnb=0,ilength=-2,NLangle=  180](A0){
$\Obs_\epsilon :$} 
    \imark(A0) \fmark[fangle=270](A0)
\node[linecolor=blue!70,fillcolor=blue!20](A3)(40,5){$q_0$}
\nodelabel[ExtNL=y,NLdist=8,AHnb=0,ilength=-2,NLangle= 180](A3){$\Obs_D :$} 

\imark(A3)
\node[linecolor=blue!70,fillcolor=blue!20](A4)(80,5){$q_1$}
\fmark[fangle=270](A4)
          
      \drawloop(A0){$\begin{array}{l} v|v, \ v \in V \\ 
c|\eps, \ c \in C \end{array}$} 
 \drawloop(A3){$a|a, \ a \in A$}
 \drawloop(A4){$\begin{array}{l} v|v, \ v \in V \\ 
c|\eps, \ c \in C \end{array}$}
            \drawedge(A3,A4){$d|d, \ d \in D$}
    \end{picture}
  \end{center}
  \caption{The Orwellian observer  $\pi_{V,D} = \Obs_\epsilon \cupdot \Obs_D$.}
\label{fig:pi_{o,d}}
\end{figure}

It has been shown in~\cite{mullins2014} that a language $L$ satisfies
\emph{intransitive non-interference (INI)} if and only if
$\varphi_{INI} = \{w \in L \ \mid \pi_{V,D}(w) \neq w \}$ is opaque in
$L$ w.r.t. the observer $\pi_{V,D}$.
\end{example}
This can be generalized as follows, showing that many non-interference like
properties reduce to opacity w.r.t. Orwellian observers.
\begin{proposition}\label{prop:incl-op}
  Let $\Obs$ be a rational idempotent function (\textit{i.e.} $\Obs^2
  = \Obs$).  Then $\Obs(L) \subseteq L$ if and only if $\varphi_{\Obs}
  = \{w \in L \ \mid \Obs(w) \neq w \}$ is opaque in $L$ for $\Obs$.
\end{proposition}
\begin{proof}
  First assume that $\Obs(L) \subseteq L$ and let $w \in
  \varphi_{\Obs}$. Then $\Obs(w) \neq w$. For $w'=\Obs(w)$, we have:
  $w' \in L$ and $\Obs(w')= \Obs^2(w)=\Obs(w)=w'$, hence $w' \notin
  \varphi_{\Obs}$. Opacity of $\varphi_{\Obs}$ follows.\\
  Conversely, assume that $\varphi_{\Obs}$ is opaque and let $w$ be an
  element of $L$. If $w \in \varphi_{\Obs}$, then there exists $w' \in
  L \setminus \varphi_{\Obs}$ such that $\Obs(w) = \Obs(w')$. Since
  $w' \notin \varphi_{\Obs}$, $\Obs(w') = w'$, hence $w'=\Obs(w) \in
  L$. Otherwise, $w \notin \varphi_{\Obs}$ implies $\Obs(w) = w \in
  L$.  In all cases, $\Obs(w) \in L$ and $\Obs(L) \subseteq L$.\qed
\end{proof}

Finally, we can state the following:

\begin{proposition}
  Given an Orwellian observer $\cal O$, deciding opacity of regular
  secrets w.r.t. ${\cal O}$ for regular languages is PSPACE-complete.
\end{proposition}

\begin{proof}
  Corollary~\ref{cor:regularcase} implies that the problem is in
  PSPACE. For the PSPACE-hardness, it suffices to observe that dynamic
  or static observers are particular Orwellian observers for which the
  problem is already PSPACE-hard.\qed
\end{proof}

In the next paragraph, we show that the observation function defined
for selective declassification is an Orwellian observer.

\subsection{Selective declassification} \label{section:SD}

Intransitive non-interference with selective declassification (INISD)
generalizes INI by allowing to each downgrading action to declassify
only a subset of confidential actions. It has recently been proposed
in~\cite{best-darondeau2012} for a class of Petri net languages (that
does not include rational languages). To formalize INISD, the alphabet
is partitioned into $A = V \cupdot C \cupdot D$ as in
example~\ref{ex:pi_{o,d}}. In addition, with each declassification
action $d \in D$ is associated a specific set $C(d) \subseteq C$ of
confidential events, with the following meaning: An occurrence of $d$
in a word $w$ declassifies all previous occurrences of actions from
$C(d)$, hence these actions are observable while other confidential
events in $C$ are not.

Let $\Sigma(D) = \{\sigma \in D^* \mid |w|_d \leq 1 \mbox{ for all } d
\in D\}$ be the set of repetition-free sequences of downgrading
actions in $D$. With any $\sigma = d_1 d_2 \ldots d_n \in \Sigma(D)$,
we associate the sets:
\begin{eqnarray*}
A_\sigma & = & V \cupdot C \cupdot \{d_1, \ldots, d_n\}\\
 W_\sigma  & =  & A_\sigma^\ast \cdot d_1 \cdot (A_\sigma \setminus \{d_1\})^\ast  
  \cdot d_2  \cdot \ldots  \cdot d_n \cdot 
  (A_\sigma \setminus \{d_1, \ldots, d_n\})^\ast \\
  V_{\sigma,i} & = & V \cup \{d_j, \ i+1 \leq j \leq n\} \cup 
\bigcup_{j=i+1}^n C(d_j), \mbox{ for every } i \in \{0, \ldots,n\} 
\end{eqnarray*}
with the convention $V_{\sigma,n}=V$, and the projections
$\pi_{\sigma,i} : A^* \rightarrow V_{\sigma,i}^*$ for every $i \in
\{0, \ldots,n\}$.

For a given $\sigma=d_1 \ldots d_n \in \Sigma(D)$, the set $W_\sigma$
contains the words $w$ in $A^\ast$ where the set of all downgrading
actions is precisely $\{d_1, \ldots, d_n\}$ and such that the last
occurrence of $d_i$ precedes the last occurrence of $d_{i+1}$ for any
$1 \leq i \leq n-1$. Note that the family of all these sets
$\{W_\sigma, \ \sigma \in \Sigma \}$ form a partition of $A^*$.
Besides, the projection $\pi_{\sigma,i}$ observes in clear any
confidential event in $\cup_{j=i+1}^nC(d_j)$, in addition to the
visible events in $V$ and the declassifying events from $\sigma$.



\smallskip Now the property called INISD in~\cite{best-darondeau2012}
can be stated in our general context for a language $L$ as follows:
For any $\sigma \in \Sigma(D)$ and for any word $w = w_0 d_1 w_1 \ldots
d_n w_n$ in $L \cap W_\sigma$, there exists a word $w' = w'_0 d_1 w'_1
\ldots d_n w'_n$ in $L\cap W_\sigma$ such that for every $i \in \{0,
\ldots,n\}$, $w'_i \in V_{\sigma,i}^*$ and $\pi_{\sigma,i}(w_i) =
\pi_{\sigma,i}(w'_i)$. We have:

\begin{proposition}\label{prop:INISD}
  The INISD property for languages in $\La$ belongs to $RIF(\La)$.
\end{proposition}
\begin{proof}
  We build an (idempotent) Orwellian observer $\Obs_{SD}$ such that a
  language $L$ satisfies INISD if and only if $\Obs_{SD}(L) \subseteq
  L$. Let $\Obs_{SD} = \bigcupdot_{\sigma \in \Sigma(D)} \Obs_\sigma$,
  where the view $\Obs_\eps$ and a generic view $\Obs_\sigma$ for some
  non empty $\sigma = d_1 \ldots d_n \in \Sigma(D)$ are depicted in
  Fig.~\ref{fig:Gsigma}.\qed
\end{proof}

\begin{figure}[ht]
  \begin{center}
\unitlength=0.8mm
\gasset{loopdiam=6,Nh=8,Nw=8,Nmr=4}
    \begin{picture}(140,32)(-18,-3)

 \node[linecolor=blue!70,fillcolor=blue!20](B0)(-20,4){$p_0$}
 \nodelabel[ExtNL=y,NLdist=8,AHnb=0,ilength=-2,NLangle=180](B0){
$\Obs_\epsilon :$} 
    \imark(B0)  \fmark[fangle=270](B0)
  \drawloop(B0){{\small$\begin{array}{l} v|v, \ v \in V \\ 
c|\eps, \ c \in C \end{array}$}}
    
\node[linecolor=blue!70,fillcolor=blue!20](A0)(20,4){$q_0$}
\nodelabel[ExtNL=y,NLdist=8,AHnb=0,ilength=-2,NLangle=180](A0){$\Obs_\sigma :$} 
\imark(A0)
 \node[linecolor=blue!70,fillcolor=blue!20](A1)(60,4){$q_1$}
                   
\drawedge(A0,A1){$d_1$}
 \node[Nframe=n,Nw=20,Nh=20,Nmr=10](A2)(95,4){$\cdots$}
\node[linecolor=blue!70,fillcolor=blue!20](A3)(130,4){$q_n$}
\fmark[fangle=270](A3)

\drawloop(A0){{\small $\begin{array}{l} v|v, \ v \in V_{\sigma,0} \\ 
c|\eps, \ c \in A_\sigma \setminus V_{\sigma,0}  \end{array}$}}
\drawloop(A1){{\small $\begin{array}{l} v|v, \ v \in V_{\sigma,1} \\ 
c|\eps, \ c \in A_\sigma \setminus V_{\sigma,1}  \end{array}$}}
 \drawloop(A3){{\small $\begin{array}{l} v|v, \ v \in V_{\sigma,n} \\ 
c|\eps, \ c \in A_\sigma \setminus V_{\sigma,n}  \end{array}$}}
    \drawedge(A1,A2){$d_2$}
    \drawedge(A2,A3){$d_n$}
    \end{picture}
  \end{center}
  \caption{Views of the observation $\Obs_{SD}$}
\label{fig:Gsigma}
\end{figure}
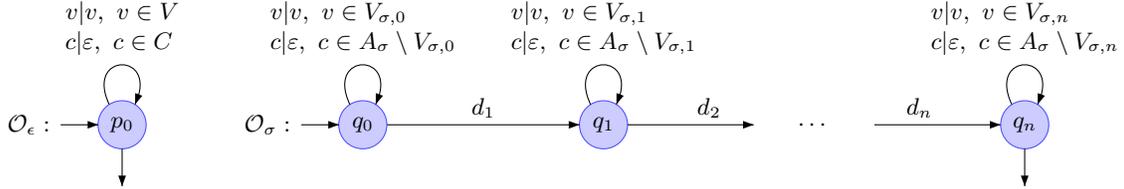

Let $w = w_0 d_1 w_1 \ldots d_n w_n$ be a word in $L \cap W_\sigma$,
the observation of $w$ is 
$$\Obs_\sigma(w) = \pi_{\sigma,0}(w_0) d_1
\pi_{\sigma,1}(w_1) \ldots d_n \pi_{\sigma,1}(w_n).$$ Then $L$
satisfies INISD if and only if $\Obs_{\sigma}(L\cap W_\sigma)
\subseteq L \cap W_\sigma$ for any $\sigma \in \Sigma(D)$.  Since the
family $\{W_\sigma, \ \sigma \in \Sigma \}$ is a partition of $A^*$,
the family $\{L\cap W_\sigma, \ \sigma \in \Sigma \}$ is a partition
of $L$ and the result follows. Each view $\Obs_\sigma$ is idempotent
and the partitionning also ensures that $\Obs_{SD}$ itself is
idempotent. As a consequence, proposition~\ref{prop:incl-op} applies
here.

\begin{remark}
  Also note that a secret $\varphi$ is opaque for a language $L$
  w.r.t. $\Obs_{SD}$ if and only if for all $\sigma \in \Sigma(D)$,
  $\varphi \cap W_\sigma$ is opaque for $L \cap W_\sigma$
  w.r.t. $\Obs_\sigma$. Indeed, the result again holds because the
  family $\{L\cap W_\sigma, \ \sigma \in \Sigma \}$ is a partition of
  $L$: for all $\sigma \in \Sigma(D)$, $\Obs_\sigma(W_\sigma)
  \subseteq W_\sigma$, we have that $\varphi$ is opaque for $L$
  w.r.t. $\Obs_{SD}$ if and only if for all $\sigma \in \Sigma(D)$,
$$\Obs_{\sigma}(\varphi \cap W_\sigma) \subseteq 
\Obs_{\sigma}((L \setminus \varphi) \cap W_\sigma) 
=  \Obs_{\sigma}((L \cap W_\sigma) \setminus (\varphi \cap W_\sigma)).$$
\end{remark}  

Like before, for regular languages, decidability of INISD as well as
opacity under $\Obs_{SD}$, are consequences of
corollary~\ref{cor:regularcase} and proposition~\ref{prop:INISD}
above. This property is studied in~\cite{best-darondeau2012} for the
prefix languages of (unbounded) labelled Petri nets. This family is
closed under intersection, inverse morphisms and alphabetical
morphisms, hence it is also closed under rational transductions (by
Nivat's theorem~\cite{berstel79}), but it has an undecidable inclusion
problem. A very nice proof is given in~\cite{best-darondeau2012} for
the decidability of the INISD property: it relies on the decidability
of the inclusion problem for the particular case of free nets (where
all transitions have distinct labels, different from $\eps$).

\begin{figure}
  \begin{center}
    \unitlength=3pt
    \begin{picture}(76, 60)(-10, -30)
    \gasset{Nw=5,Nh=5,Nmr=2.5,linewidth=0.20,curvedepth=0}
         \node[linecolor=blue!70,fillcolor=blue!20](B0)(0,0){$p_{i3}$}
                    \node[linecolor=blue!70,fillcolor=blue!20](B1)(0,15){$p_{i2}$}
                       \imark(B1)
                                      \nodelabel[ExtNL=y,NLdist=6,AHnb=0,ilength=-2,NLangle= 180](B1)
                                      {$\mbox{Goat}(i):$}
                                        \node[linecolor=blue!70,fillcolor=blue!20](B3)(15,15){$p_{i4}$}
                        \node[linecolor=blue!70,fillcolor=blue!20](B2)(0,30){$p_{i1}$}
                        
       \drawedge(B0,B3){$\overline{d_3}$}    
             \drawedge(B2,B3){$\overline{d_1}$}     
                          \drawedge[ELside=r](B2,B1){$l_2$}     
                   \drawedge(B1,B3){$\overline{d_2}$} 
                                       \drawedge[ELside=r](B1,B0){$l_3$}  
                                       \gasset{curvedepth=22}
                                               \drawedge(B0,B2){$l_1$}  
    \node[linecolor=blue!70,fillcolor=blue!20](A1)(42,30){$q_{j1}$}
    \node[linecolor=blue!70,fillcolor=blue!20](A0)(55,15){$q_{j0}$}
        \imark(A0)
                          \nodelabel[ExtNL=y,NLdist=6,AHnb=0,ilength=-2,NLangle=  180](A0){$\mbox{Raptor}(j):$}
        \node[linecolor=blue!70,fillcolor=blue!20](A3)(55,0){$q_{j3}$}
        \node[linecolor=blue!70,fillcolor=blue!20](A2)(68,30){$q_{j2}$}
        
        \gasset{curvedepth=3}
        \drawedge(A0,A1){$h_1$}
        \drawedge(A1,A0){$d_1$}
                \drawedge(A0,A2){$h_2$}
        \drawedge(A2,A0){$d_2$}
                        \drawedge(A0,A3){$h_3$}
        \drawedge(A3,A0){$d_3$}
        
          \node[linecolor=blue!70,fillcolor=blue!20](C0)(-20,-15){$r_{10}$}
                    \node[linecolor=blue!70,fillcolor=blue!20](C1)(-20,-30){$r_{11}$}
                    \imark(C1)
                                            \nodelabel[ExtNL=y,NLdist=6,AHnb=0,ilength=-2,NLangle= 180](C1){$\mbox{Gate}(1):$} 
                                            
                       \node[linecolor=blue!70,fillcolor=blue!20](C2)(30,-15){$r_{20}$}
                    \node[linecolor=blue!70,fillcolor=blue!20](C3)(30,-30){$r_{21}$}
                    \imark(C2)
                                            \nodelabel[ExtNL=y,NLdist=6,AHnb=0,ilength=-2,NLangle= 180](C2){$\mbox{Gate}(2):$} 
                                            
                         \node[linecolor=blue!70,fillcolor=blue!20](C4)(78,-15){$r_{30}$}
                    \node[linecolor=blue!70,fillcolor=blue!20](C5)(78,-30){$r_{31}$}
                                            \nodelabel[ExtNL=y,NLdist=6,AHnb=0,ilength=-2,NLangle= 180](C5){$\mbox{Gate}(3):$} 
                                                  \imark(C5)
                                                  
    \gasset{curvedepth=3}
        \drawedge(C0,C1){$\overline{h_1}$}
        \drawedge(C1,C0){$\overline{h_3}$}
         \drawloop[loopangle=0](C1){$\overline{l_1}, \overline{h_1}$}
                  
     \gasset{curvedepth=3}
        \drawedge(C2,C3){$\overline{h_2}$}
        \drawedge(C3,C2){$\overline{h_1}$}
         \drawloop[loopangle=0](C3){$\overline{l_2}, \overline{h_2}$}

     \gasset{curvedepth=3}
        \drawedge(C4,C5){$\overline{h_3}$}
        \drawedge(C5,C4){$\overline{h_2}$}
         \drawloop[loopangle=0](C5){$\overline{l_3}, \overline{h_3}$}

    \end{picture}
  \end{center}
    \caption{The dining Raptors}
\label{fig:DR}
\end{figure}
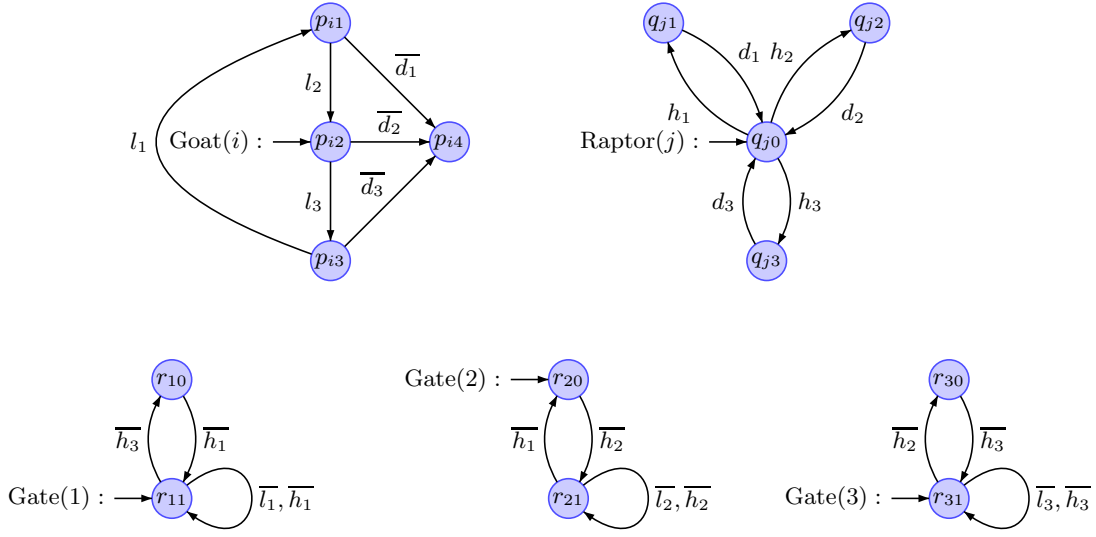

\smallskip
The following example (inspired from~\cite{best-darondeau2012}) tries
to explain the essence of selective declassification.
\begin{example}[The Dining Raptors]
  A circuit followed by a herd of goats is divided in three
  sections. Each section is guarded by a gate. When gate $i$ is open,
  goats can move clockwise from section $i$ to section $i+1
  \pmod{3}$. The center of the circuit is occupied by a den of
  raptors. When gate $i+1 \pmod{3}$ is open, a raptor can leave the
  den and hide around gate $i$ after opening it and closing gate $i+1
  \pmod{3}$ to increase chance of success. When a raptor is embushed
  near a section and there is a goat in this section, the raptor can
  catch prey and come back to the den.

This scenario is modelled with the transition system
\begin{eqnarray*}
\mbox{DR}(n,m) & = &  \prod_{i=1}^n \mbox{Goat}(i) \times  \prod_{j=1}^m\mbox{Raptor}(j) \times \prod_{k=1}^3 \mbox{Gate}(k) 
 \end{eqnarray*}
 obtained by synchronizing the components depicted
 in~Figure~\ref{fig:DR} on the complementary actions. Goats' move from
 gate $i$ to gate $i+1 \pmod{3}$ is modelled with visible action
 $l_1$, raptors' embush action at section $i$, with the confidential
 action $h_i$ and the raptors' catch action in section $i$, by the
 declassification action $d_i$.  Opacity of $\varphi_{DR}$
 w.r.t. ${\cal O}_{SD}$ in $L(DR(m, n))$ where
    \begin{eqnarray*}
 \varphi_{DR} & = & \{u \in L(DR(m, n)) \mid  {\cal O}_{SD}(u) \not = u\}
 \end{eqnarray*}
 comes down to absence of information the goats can get from
 environment about the moment they will be caught until this
 happens. Hence there is no strategy that they can oppose to the
 raptors. In the case where initially goats are in section $2$ and
 gates $1$ and $3$ are opened, as shown in~Figure~\ref{fig:DR},
 $L(\mbox{DR}(n,m))$ is not opaque w.r.t. $ \varphi_{DR}$ since
 $l_3l_1h_2l_2$ reveals the secret ($h_2l_3l_1l_2$, $l_3h_2l_1l_2$ and
 $l_3l_1h_2l_2$ are the only traces observed as $l_3l_1l_2$) and this,
 for any number of raptors and goats. This example may be of course
 modified in various ways as follows. If all three gates are open,
 goat $1$ never realizes she dies since $l_3l_1h_1d_1$ does not
 reveals the secret but following this, as gate $2$ is now close, goat
 $2$ after $l_3l_1l_2$ will know that a raptor is embushed at gate $2$
 since $l_3l_1h_1d_1l_3l_1h_2l_2$ reveals the secret. If only gate $3$
 is open, $l_3h_2d_2h_1l_1$ reveals to the herd, that one of them is
 now trapped in section $2$. Finally, if we dismantle all three gates,
 the only synchronizing actions are now the declassification ones and
 $\varphi_{DR}$ becomes opaque w.r.t. ${\cal O}_{SD}$.
%
\end{example}

\section{RIF properties with full rational relations} \label{sec:ex}
In this section, we first revisit Basic Security Predicates (BSP)
presented in~\cite{Mantel00,Mantel01} and used as building blocks of
the Mantel's generic security model. In the second part, we
investigate anonymity properties.

\subsection{Basic Security Predicates} \label{sec:bsp}

For BSPs, the alphabet $A$ is partitioned into $A = V \cupdot C
\cupdot N$, where $V$ is the set of visible events, $C$ is the set of
confidential events and $N$ is a set of internal events. Informally, a
BSP for a given language $L$ over $A$, is an implication stating that
for any word $w$ in $L$ satisfying some restriction condition, there
exists a word $w'$ also in $L$ which is observationnally equivalent to
$w$ and which fulfills some closure condition describing the way $w'$
is obtained from $w$ by adding or removing some confidential events.
The conditions are sometimes parametrized by an additional set $X
\subseteq A$ of so-called admissible events. We prove:
\begin{proposition}\label{prop:BSP}
  Any BSP over languages in some family $\La$ belongs to $RIF(\La)$.
\end{proposition}
\begin{proof}
The proof mainly consists in exhibiting rational observers together
with an inclusion relation such that a language $L$ satisfies a given
BSP if and only if this relation holds.  We give the general idea with
several examples.  For any $B \subseteq A$, we write $\overline{B}= A
\setminus B$. First observe that, starting from some inclusion
relation $\Obs_1(L) \subseteq \Obs_2(L)$ for rational observers
$\Obs_1$ and $\Obs_2$, ignoring events from $N$ reduces to composing
both sides with $\pi_{\overline{N}}$. This is simply done by adding
loops labeled by $n |\eps$, for all $n\in N$, on all states of the
tranducers realizing $\Obs_1$ and $\Obs_2$ over $V \cupdot C$. This
operation corresponds to variants of the properties.

\begin{enumerate}
\item The simplest predicate called \emph{Strict Removal of events
    (SR)} corresponds to a projection: $L$ satisfies $SR$ if
  $\pi_{\overline{C}}(L) \subseteq L$. The non strict variant \emph{(R)}
  where events of $N$ are ignored corresponds to $\pi_{V}(L) \subseteq
  \pi_{\overline{N}}(L)$, since the composition of
  $\pi_{\overline{C}}$ and $\pi_{\overline{N}}$ produces $\pi_{V}$.
\item We now turn to predicates with stepwise deletion of events. A
  language $L$ satisfies $SD$ (\emph{Strict Deletion of events}) if
  for any $w=w_1 c w_2 \in L$, with $c\in C$ and $\pi_C(w_2)=\eps$,
  then $w_1 w_2 \in L$. As noted in~\cite{DSouzaHRS11}, this property
  is equivalent to $l\mbox{-}del(L) \subseteq L$, where $l\mbox{-}del$
  is the function associating with a word $w$ the word
  $l\mbox{-}del(w)$ obtained from $w$ by deleting the last
  confidential event.  This function is realized by the transducer in
  Figure~\ref{fig:ex2} left.  The observation itself is described
  in~\cite{Mantel00} as a recursive operation: starting from $w= w_0
  c_1 w_1 c_2 \ldots w_{p-1} c_p w_{p}$ with $w_i \in \overline{C}^*$
  for $0\leq i \leq p$, the words obtained by successively removing
  all confidential actions from the right to the left of $w$ must also
  belong to $L$. This corresponds to applying the star operation to
  $l\mbox{-}del$ (for the composition of relations), resulting in
  $\Obs_{del} = \cup_{k \geq 0}l\mbox{-}del ^k$, which is not a
  function. While the star operation does not necessarily preserve
  rational relations~\cite{sakarovitch09}, in this case, $\Obs_{del}$
  is realized by the transducer in Figure~\ref{fig:ex2} right.

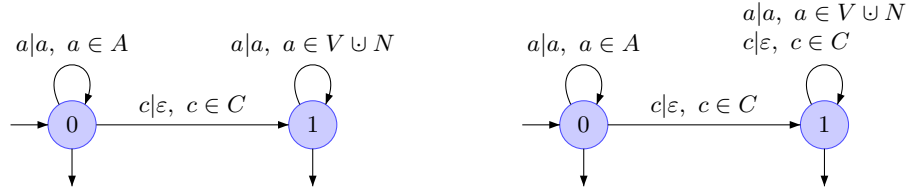
\begin{figure}[htbf]
\unitlength=0.8mm
\gasset{loopdiam=6,Nh=8,Nw=8,Nmr=4}
\begin{picture}(120,25)(-25,1)

\node[linecolor=blue!70,fillcolor=blue!20](q0)(10,8){$0$}
\node[linecolor=blue!70,fillcolor=blue!20](q1)(50,8){$1$}
\imark(q0) \fmark[fangle=270](q0) \fmark[fangle=270](q1)

\drawloop(q0){$a|a, \ a \in A$}

\drawedge(q0,q1){$c|\eps, \ c \in C$} 

\drawloop(q1){$a|a, \ a \in V \cupdot N$}

\put(85,0){
\node[linecolor=blue!70,fillcolor=blue!20](q0)(10,8){$0$}
\node[linecolor=blue!70,fillcolor=blue!20](q1)(50,8){$1$}
\imark(q0) \fmark[fangle=270](q0) \fmark[fangle=270](q1)

\drawloop(q0){$a|a, \ a \in A$}

\drawedge(q0,q1){$c |\eps, \ c \in C$} 
\drawloop(q1){$\begin{array}{l} a|a, \ a \in V \cupdot N \\ 
c|\eps, \ c \in C \end{array}$}

}

\end{picture}
\caption{Transducers for Strict Deletion $l\mbox{-}del$ (left) and
  observation $\Obs_{del}$ (right)}
\label{fig:ex2}
\end{figure}
Again the non strict variant $(D)$ is obtained by composition on both
sides with $\pi_{\overline{N}}$. An other variant, \emph{Backward
  Strict Deletion of confidential events (BSD)}, is defined by: $L$
satisfies $BSD$ if for any $w=w_1 c w_2 \in L$ with $\pi_C(w_2)=\eps$,
there is $w'_2$ such that
$\pi_{\overline{N}}(w_2)=\pi_{\overline{N}}(w'_2)$ and $w_1 w'_2 \in
L$. In this case, only the suffixes $w_2$ and $w'_2$ following the
last confidential event can differ on internal events from $N$. The
corresponding observation relation $\Obs_{BSD}$ is defined by
associating with a word $w=w_1 c w_2$ such that $\pi_C(w_2)=\eps$, all
words obtained from $w$ by removing $c$ and replacing $w_2$ by some
$w'_2$ such that
$\pi_{\overline{N}}(w_2)=\pi_{\overline{N}}(w'_2)$. Then, the rational
observation $\Obs_{BSD}$ realized by the transducer in
Fig.~\ref{fig:ex3} left (which is not a function) is such that $L$
satisfies $BSD$ if and only if $\Obs_{BSD}(L) \subseteq L$.

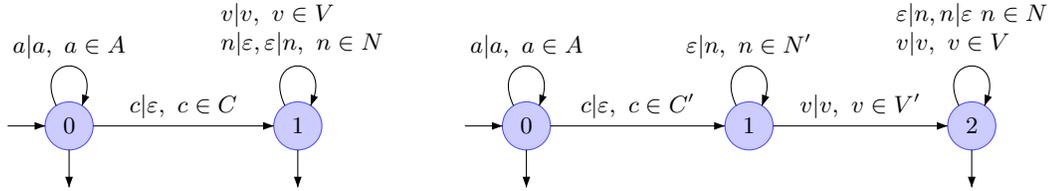
\begin{figure}[htbf]
\unitlength=0.8mm
\gasset{loopdiam=6,Nh=8,Nw=8,Nmr=4}
\begin{center}
\begin{picture}(120,24)(20,5)
\node[linecolor=blue!70,fillcolor=blue!20](q0)(10,8){$0$}
\node[linecolor=blue!70,fillcolor=blue!20](q1)(48,8){$1$}
\imark(q0) \fmark[fangle=270](q0) \fmark[fangle=270](q1)
\drawloop(q0){$a|a, \ a \in A$}

\drawedge(q0,q1){$c|\eps, \ c \in C$} 

\drawloop(q1){
$\begin{array}{l} v|v, \ v \in V \\ n|\eps, \eps |n, \ n \in N \end{array}$}

\put(76,0){
\node[linecolor=blue!70,fillcolor=blue!20](q0)(10,8){$0$}
\node[linecolor=blue!70,fillcolor=blue!20](q1)(47,8){$1$}
\node[linecolor=blue!70,fillcolor=blue!20](q2)(84,8){$2$}
\imark(q0) \fmark[fangle=270](q0) \fmark[fangle=270](q2)

\drawloop(q0){$a|a, \ a \in A$}
\drawedge(q0,q1){$c |\eps, \ c \in C'$} 
\drawloop(q1){$\eps|n, \ n \in N'$}
\drawedge(q1,q2){$v |v, \ v \in V'$} 
\drawloop(q2){$\begin{array}{l} \eps|n, n| \eps\  n \in  N \\ 
v|v, \ v \in V \end{array}$}

}
\end{picture}
\end{center}
\caption{Transducers for BSD (left) and FCD (right)}
\label{fig:ex3}
\end{figure}

A last variant, called \emph{Forward Correctable Deletion (FCD)}
in~\cite{DSouzaHRS11} considers fixed subsets $V' \subseteq V$, $C'
\subseteq C$ and $N' \subseteq N$. A language $L$ satisfies $FCD$ if
for any $w= w_1 c v w_2$ in $L$ with $c \in C'$, $v\in V'$ and
$\pi_C(w_2)=\eps$, there are some $w' \in N'^*$ and $w'_2$ such that
$\pi_{\overline{N}}(w_2)=\pi_{\overline{N}}(w'_2)$ and $w_1 w' v w'_2
\in L$. The corresponding transducer $\Obs_{FCD}$ is depicted in
Fig.~\ref{fig:ex3} right, with $L$ satisfies $FCD$ if and only if
$\Obs_{FCD}(L) \subseteq L$.

\item Finally, the last class concerns stepwise insertion of events. A
  language $L$ satisfies $SI$ (\emph{Strict insertion of events}) if
  for any $w=w_1 w_2 \in L$, with $\pi_C(w_2)=\eps$, and for any $c
  \in C$, we have $w_1 c w_2 \in L$. The corresponding relation
  $l\mbox{-}ins$ (which is also not a function) such that $L$
  satisfies $SI$ iff $l\mbox{-}ins(L) \subseteq L$, is realized by a
  transducer similar to the one in Fig.~\ref{fig:ex2} left, where the
  middle labels $c|\eps, \ c \in C$ are replaced by $\eps|c, \ c \in
  C$. The non strict $(I)$, the backward $(BSI)$ and the correctable
  $(FCI)$ variants are obtained similarly.

  \smallskip The remaining four predicates concern insertion with
  respect to admissible events.  For a given subset $X$ of $A$, a
  language $L$ satisfies \emph{Strict insertion of $X$-admissible
    events ($SIA^X$}) if for any $w=w_1 w_2 \in L$ such that
  $\pi_C(w_2)=\eps$ and there are some $w_3 \in A^*$ and $c\in C$ with
  $w_3c \in L$ and $\pi_X(w_1)=\pi_X(w_3)$, we have $w_1 c w_2 \in
  L$. In this case, recall that the left quotient of language $M'$ by
  language $M$ is defined by $M^{-1}M'= \{w_1 \in A^* \ \mid \ w_1 w_2
  \in M' \mbox{ for some } w_2 \in M\}$. For a fixed regular language
  $M$, the left quotient by $M$ and the concatenation by $M$ are
  rational relations~\cite{berstel79}. For each $c \in C$, we consider
  the two following rational relations:
\begin{itemize}
\item $l\mbox{-}ins_c$ is the variant of $l\mbox{-}ins$ where the
  single fixed letter $c$ is inserted,
\item $\Obs_c^X$ is defined by $\Obs_c^X(u)=\pi_X^{-1}
  (\pi_X(c^{-1}u)).c.(V \cupdot N)^*$ for $u \in A^*$.
\end{itemize}
Then $L$ satisfies  $SIA^X$  if and only 
$\bigcup_{c \in C} (l\mbox{-}ins_c( L) \cap \Obs_c^X(L)) \subseteq L$.

Similar relations hold for the variants $(IA^X)$, $(BSIA^X)$ and
$(FCIA^X)$, hence these four cases are slightly different from the
previous ones. 
\end{enumerate} 
\qed \end{proof}

\medskip In~\cite{DSouzaHRS11}, the decidability results for all 14
BSPs on regular languages are obtained by ad-hoc proofs establishing
that regularity is preserved by the various $op_1$, $op_2$
operations. These include auxiliary functions on languages (like
$mark$, $unmark$, etc.) that are unnecessary in our setting. Actually,
we show how decidability of BSPs is an immediate consequence of
corollary~\ref{cor:regularcase} and proposition~\ref{prop:BSP} above.
The more difficult case of pushdown systems (generating prefix-closed
context-free languages) is also investigated in~\cite{DSouzaHRS11}:
Although context-free languages are closed under rational
transductions, they are not closed under intersection and the
inclusion problem is undecidable for context-free
languages~\cite{berstel79}. Finally, several undecidability results
are presented in~\cite{DSouzaHRS11}. In particular, they exhibit an
information flow property called Weak Non Inference (WNI) shown to be
undecidable even for regular languages. Hence, WNI cannot be expressed
neither as a conjunction of BSPs, and as matter of fact, neither as an
RIFP. Also, in order to get decidable cases, authors had to restrict
the languages and/or the class of properties like reducing the size of
the alphabet ($card(V) \leq 1$ and $card(C) \leq 1$).

\subsection{Conditional anonymity}\label{sec:CA}

Conditional or escrowed anonymity is concerned with the revocation of
the guarantee, under well-defined conditions, to which an agent
agrees, that his identification w.r.t. a particular (non-secret)
action will remain secret and in such case, conditional anonymity
guarantees the unlinkability of revoked users in order to guarantee
anonymity to ``legitimate'' agents~\cite{danezis2008}. As suggested
in~\cite{Mazare2008}, Orwellian observation can be used to model {\em
  conditional anonymity} but~\cite{Mazare2008} contains neither a
definition of such a property, nor any investigation of its
decidability. We close the gap in this paper.
 
The alphabet is partitioned into $A= V \cupdot P \cupdot R$ where $P$
is the set of actions performed by anonymous participants, $V$ is the
set of visible actions and $R$ is the set of anonymity revocation
actions, such that for each participant corresponds a dedicated
revocation action $r$ allowing to reveal the subset $P(r)$ of all its
anonymous actions. Hence the sets $P(r)$ are mutually disjoint.

In~\cite{Schneider96}, definitions of {\em weak and strong anonymity}
are given in the setting of the process algebra CSP.  A language is
\emph{strongly anonymous (SA)} if it is stable under any
``perturbation'' of anonymous actions where an anonymous action in $P$
can be replaced by any other element of $P$.  It is \emph{weakly
  anonymous (WA)} if it is stable under any permutation on the set of
anonymous actions. For a finite set $Z$, we denote by $S_Z$ the set of
all permutations on $Z$.  We first have:
\begin{proposition} \label{prop:AinRIF}
Weak and strong anonymity on languages in $\La$ belong to $RIF(\La)$.
\end{proposition} 
\noindent \emph{Proof.} 
  For these two properties, the subalphabet $R$ of revocation actions
  is empty. We express the properties in our language-based setting,
  similarly as in~\cite{Mazare2008}.

  A language $L$ is strongly anonymous w.r.t.\ $P$ if 
  $\Obs_{SA}^P(L) \subseteq L$ where $\Obs_{SA}^P$ is the mapping
  defined on $A = V \cupdot P$ by: $\Obs_{SA}^P(a) = P$ if $a \in P$
  and $\Obs_{SA}^P(a) = \{a\}$ otherwise. Such mappings (called
  rational substitutions in~\cite{berstel79}) are well known to be
  rational relations, hence the result follows.

  \smallskip A language $L$ is weakly anonymous w.r.t.\ $P$ if
  $\Obs_{WA}^P(L) \subseteq L$ where $\Obs_{WA}^P = \bigcupdot_{\alpha
    \in S_{P}} \Obs_\alpha$ and $\Obs_\alpha$ is the morphism which
  applies the permutation $\alpha$ on letters of $P$: 

  $\Obs_\alpha(a) = \alpha(a)$ if $a \in P$ and $\Obs_\alpha(a) = a$
  otherwise.

With any $\sigma \subseteq R$, we associate:
\begin{itemize}
\item $W_\sigma= \{w \in A^* \mid \ \pi_R (w) \in \sigma^* \}$, the
  set of words $w$ in $A^\ast$ where the set of revocation actions
  appearing in $w$ is $\sigma$,
\item $P_\sigma = P \setminus \bigcupdot_{r \in \sigma}P(r)$, the set
  of actions of legitimate agents.
\end{itemize}

We denote by $2^R$ the powerset of $R$ and remark that here also, the
sets $W_\sigma$ for $\sigma \in 2^R$ form a partition of $A^*$. In
order to provide at any moment strong (weak) anonymization to
legitimate agents, we define conditional anonymity as follows:

\begin{definition}
  With the notations above, a language $L$ on $V \cupdot P \cupdot R$
  is:
\begin{itemize}
\item {\em conditionally weakly anonymous} (CWA) if for any $\sigma
  \subseteq R$, $\ L \cap W_\sigma$ is WA w.r.t.\ $P_\sigma$,
\item {\em conditionally strongly anonymous} (CSA) if for any $\sigma
  \subseteq R$, $\ L \cap W_\sigma$ is SA w.r.t.\ $P_\sigma$.
\end{itemize}
\end{definition}

Now we have:
\begin{proposition}\label{prop:CAinRIF} 
  Weak and strong conditional anonymity on languages in $\La$ belong
  to $RIF(\La)$.
\end{proposition} 
\begin{proof}
  We build rational observers, with a view-like component for each
  possible subset $\sigma$ of revoked users, corresponding to
  $\Obs_{SA}$ (resp. $\Obs_{WA}$) localized to $W_\sigma$,
  \textit{i.e.}  revocation actions are those in $\sigma$, anonymous
  actions are restricted to $P_\sigma$ and visible actions are
  extended to $V \cupdot \bigcupdot_{r \in \sigma}P(r)$:

 \[\Obs_{CSA} = \bigcupdot_{\sigma \in 2^{R}} \Obs_{SA}^{P_\sigma} \mbox{ and }
 \Obs_{CWA} = \bigcupdot_{\sigma \in 2^{R}} \Obs_{WA}^{P_\sigma}
\]

Then $L$ is {\em conditionally strongly anonymous} (CSA) if and only
if $\Obs_{CSA}(L) \subseteq L$ and $L$ is {\em conditionally weakly
  anonymous} (CWA) if and only if $\Obs_{CWA}(L) \subseteq L$, which
yields the result.\qed
\end{proof}

\section{Related works.}\label{sec:relw}
Along the lines, important connections between RIFPs and information
flow properties have been established, hence in this section, we will
focus on extending the picture.


Algorithms for verifying opacity in Discrete Event Systems
w.r.t. projections are presented together with applications in
\cite{Darondeau2007,Takai2009,Hadji2011,Lin2011}.
In~\cite{Darondeau2007}, the authors consider a concurrent version of
opacity and show that it is decidable for regular systems and secrets.
In~\cite{Takai2009}, the authors define what they called {\em secrecy
} and provide algorithms for verifying this property. A system
property satisfies {\em secrecy} if the property and its negation are
state-based opaque. In~\cite{Lin2011} the author provides an algorithm
for verifying state-based opacity (called {\em strong opacity}) and
shows how opacity can be instantiated to important security properties
in computer systems and communication protocols, namely anonymity and
secrecy.  In~\cite{Hadji2011}, the authors define the notion of
K-step opacity where the system remains state-based opaque in any step
up to depth-k observations that is, any observation disclosing the
secret has a length greater than k. Two methods are proposed for
verifying K-step opacity.  All these verification problems can be
uniformly reduced to the RIFP verification problem.


In~\cite{focardi01classification}, the authors provide decision
procedures for a large class of trace-based security properties that
can all be reduced to the RIFP verification problem for regular
languages. In~\cite{Meyden2007b}, decision procedures are given for
trace-based properties like non-deducibility, generalized
non-interference and forward correctability. The PSPACE-completeness
results for these procedures can be reduced to our results.

Concerning intransitive information flow (IIF), {\em non-interference}
(NI) and {\em intransitive non-interference} (INI) for deterministic
Mealy machines have been defined
in~\cite{rushby:channel.control}. In~\cite{Pinsky1995}, an algorithm
is provided for INI.  A formulation of INI in the context of
non-deterministic LTSs is given in~\cite{mullins00ai}, in the form of
a property called {\em admissible interference} (AI), which is
verified by reduction to a stronger version of NI. This property,
called {\em strong non-deterministic non-interference} (SNNI)
in~\cite{focardi01classification}, is applied to $N$ finite transition
systems where $N$ is the number of downgrading transitions of the
original system. This problem was also reduced to the opacity
verification problem w.r.t. Orwellian projections
in~\cite{mullins2014}.  In~\cite{Bossi04modellingdowngrading}, various
notions of trace-based INI declassification properties are considered
and compared. In contrast, our generic model is instantiable to a much
larger class of IIF properties.

In~\cite{Meyden2007}, the author has argued that Rushby's definition
of security for intransitive policies suffers from some flaw, and
proposed some stronger variations. The considered flaw relies to the
fact that, if $u \in W_{d_1}$ and $v \in W_{d_2}$, that is $u$
(resp. $v$) declassifies only $h_1 \in H(d_1)$ (resp. $h_2 \in
H(d_2)$), then the shuffle of $u$ and $v$ resulting of their
concurrent interaction will reveal the order in which $h_1$ and $h_2$
have been executed. The proof techniques used in this paper for
deciding the RIFP verification problem relies on their end-to-end
execution semantics and hence does not address this problem.






\section{Conclusion}\label{sec:conc}

In this paper we have introduced a language-theoretic model for
trace-based information flow properties, the RIFPs where observers are
modelled by rational transducers. Given a family $\La$ of languages,
our model provides a generic decidability result to the $RIF(\La)$
verification problem: Given $L \in \La$ and a security property $P$ in
RIF($\La$), does $L$ satisfy $P$? When $\La$ is the class ${\cal R}eg$
of regular languages, the problem is shown PSPACE-complete. This
result subsumes most decidability results for finite systems. In order
to prove that, we have shown that opacity properties and Mantel's
BSPs, two major generic models for trace-based IF properties, are
RIFPs.  We have illustrated the expressiveness of our model by showing
that the verification problem of INI and INISD can be reduced to the
verification problem of opacity w.r.t. a subclass of rational
observers called rational Orwellian observers. Finally we have
illustrated the applicability of our framework by providing the first
formal specification of {\em conditional anonymity}.

As far as we know, the only decidability
results of trace-based security properties for infinite systems are
presented~in~\cite{best-darondeau-gorrieri2011,best-darondeau2012,DSouzaHRS11}.
Hence, the approaches of the present paper and
\cite{best-darondeau2012,DSouzaHRS11} lead to the question (which is
so far open, to the best of our knowledge) of which infinite systems
have a decidable verification problem for BSPs. 
%
Another line for future work would be to investigate the links between
our framework and the logics studied in~\cite{SecLTL2012}
and~\cite{ClarksonFKMRS14}.

\newcommand{\etalchar}[1]{$^{#1}$}



\end{document}